\documentclass[journal,onecolumn]{IEEEtran}

\usepackage{amsmath}
\usepackage{amssymb}
\usepackage{amsthm}
\usepackage{multirow}
\newtheorem{theorem}{Theorem}

\newtheorem{lemma}[theorem]{Lemma}

\newcommand{\bb}{\mathbb}

\newcommand{\cC}{\mathcal{C}}
\newcommand{\cD}{\mathcal{D}}
\newcommand{\cL}{\mathcal{L}}

\newcommand{\be}{\beta}
\newcommand{\de}{\delta}
\newcommand{\De}{\Delta}
\newcommand{\ga}{\gamma}
\newcommand{\la}{\lambda}
\newcommand{\ep}{\epsilon}

\newcommand{\zp}{\zeta_{p}}

\newcommand{\pr}{\prime}

\newcommand{\zs}{\{0\}}
\newcommand{\sm}{\setminus}

\newcommand{\supp}{\text{supp}}

\newcommand{\ul}{\underline}

\newcommand{\lan}{\langle}
\newcommand{\ran}{\rangle}

\newcommand{\wti}{\widetilde}

\newcommand{\Tr}{\text{Tr}}
\newcommand{\TQq}{\Tr^{Q}_{q}}

\newcommand{\Tqp}{\Tr^{q}_{p}}
\newcommand{\su}{\succ}

\newcommand{\F}{\mathbb{F}}
\newcommand{\Fp}{\mathbb{F}_p}
\newcommand{\Fq}{\mathbb{F}_q}
\newcommand{\FQ}{\mathbb{F}_Q}

\newcommand{\ux}{\underline{x}}

\newcommand{\ub}{\underline{b}}
\newcommand{\uy}{\underline{y}}
\newcommand{\subsp}{\left[{\FQ^t \atop r}\right]}

\newcommand{\tN}{\widetilde{N}}
\newcommand{\dimq}{\dim_{\Fq}}
\newcommand{\Hp}{H^{\perp}}

\newcommand{\tw}{\wti{w}}

\newcommand{\wHr}{\widetilde{H}_r}
\newcommand{\wWh}{\widetilde{W}_h}

\newcommand{\Hrp}{H_r^{\perp}}
\newcommand{\fm}{\frac{m}{2}}
\newcommand{\Fqm}{\F_{q^{\fm}}}
\newcommand{\rp}{r^{\pr}}

\newcommand{\Ls}{\cL_s}
\newcommand{\Lrp}{\cL_{r^{\pr}}}

\newcommand{\CN}{C^{(N,Q)}}
\newcommand{\CNq}{C^{(N,q)}}

\newcommand{\Ct}{C^{(2,Q)}}
\newcommand{\etaN}{\eta^{(N,Q)}}
\newcommand{\etaNq}{\eta^{(N,q)}}
\newcommand{\etaNg}{\etaN_{g^h\sum_{j=1}^tb_j\be_j^h}}
\newcommand{\etang}{\eta^{(1,Q)}_{g^h\sum_{j=1}^tb_j\be_j^h}}
\newcommand{\etaNy}{\etaN_{y_h}}
\newcommand{\etao}{\eta^{(1,Q)}}

\begin{document}

\title{The Weight Hierarchy of Some Reducible Cyclic Codes}

\author{Maosheng Xiong, Shuxing Li and Gennian Ge
\thanks{M. Xiong is with the Department of Mathematics, Hong Kong University of Science and Technology, Clear Water Bay, Kowloon, Hong Kong (e-mail: mamsxiong@ust.hk).}
\thanks{S. Li is with the Department of Mathematics, Zhejiang University, Hangzhou 310027, Zhejiang, China (e-mail: sxli@zju.edu.cn).}
\thanks{G. Ge is with the School of Mathematical Sciences, Capital Normal University, Beijing 100048, China (e-mail: gnge@zju.edu.cn). He is also with Beijing Center for Mathematics and Information Interdisciplinary Sciences, Beijing, 100048, China.}
\thanks{The research of M. Xiong was supported by RGC grant number 606211 and 609513 from Hong Kong. The research of G. Ge was supported by the National Natural Science Foundation of China under Grant No. 61171198 and Grant No. 11431003, the Importation and Development of High-Caliber Talents Project of Beijing Municipal Institutions, and Zhejiang Provincial Natural Science Foundation of China under Grant No. LZ13A010001.}
}

\maketitle

\begin{abstract}
The generalized Hamming weights (GHWs) of linear codes are fundamental parameters, the knowledge of which is of great interest in many applications. However, to determine the GHWs of linear codes is difficult in general. In this paper, we study the GHWs for a family of reducible cyclic codes and obtain the complete weight hierarchy in several cases. This is achieved by extending the idea of \cite{YLFL} into higher dimension and by employing some interesting combinatorial arguments. It shall be noted that these cyclic codes may have arbitrary number of nonzeroes.
\end{abstract}


\section{Introduction} \label{sec1}

Let $\cC$ be an $[n,k]$ linear code over the finite field $\Fq$ of order $q$, that is, $\cC \subset \Fq^n$ is a $k$-dimensional vector space over $\Fq$. For any (linear) subcode $\cD \subset \cC$, the {\em support} of $\cD$ is defined to be
$$
\supp(\cD)=\{i: 0 \le i \le n-1, c_i \ne 0 \mbox{ for some } (c_0,c_1,\ldots,c_{n-1}) \in \cD \}.
$$
For $1 \le r \le k$, the $r$-th {\em generalized Hamming weight} (GHW) of $\cC$ is given by
$$
d_r(\cC)=\min \left\{|\supp(\cD)|: \cD \subset \cC \mbox{ and }  \dimq(\cD)=r\right\}.
$$
Here $|\supp(\cD)|$ denotes the cardinality of the set $\supp(\cD)$. The set $\{d_r(\cC): 1 \le r \le k\}$ is called the {\em weight hierarchy} of $\cC$. Note that $d_1(\cC)$ is just the minimum distance of $\cC$.

Including the minimum distance, the GHWs of linear codes provide additional fundamental information. The notion was introduced by Helleseth, Kl{\o}ve, and Mykkeltveit \cite{HKM,K1} and was first used by Wei \cite{Wei} in cryptography to fully characterize the performance of linear codes when used in a wire-tap channel of type II \cite{OW} or as a $t$-resilient function. The GHWs can be used to characterize the performance of certain secret sharing schemes based on linear codes \cite{KU}. Recently, the concept of GHWs has been extended to linear network codes \cite{NYZ}, where they were used to characterize the security performance of linear network codes in a wiretap network \cite{RSS}. Apart from these cryptographic applications, the GHWs also provide detailed structural information of linear codes, the knowledge of the GHWs can be used to
\begin{itemize}
\item[1).] compute the state and branch complexity profiles of linear codes \cite{F,KTFL};
\item[2).] indicate efficient ways to shorten linear codes \cite{HK1};

\item[3).] derive subtle upper bounds on the covering radius of linear codes \cite{JL2};
\item[4).] determine the erasure list-decodability of linear codes \cite{G};
\item[5).] provide bounds on the list size in the list decoding of certain codes \cite{GGR}, etc.
\end{itemize}
In conclusion, the GHWs of linear codes provide fundamental parameters of linear codes which are important in many applications.

The study of GHWs has attracted considerable attention in the past two decades, and many results have been obtained in the literature. For example, general lower bounds and upper bounds on GHWs were derived \cite{ABL,CLZ,HKLY,HKY,Wei}, an efficient algorithm to compute the GHWs of cyclic codes was proposed \cite{JL1}, and the GHWs have been determined or estimated for many series of linear codes such as Hamming codes \cite{Wei}, Reed-Muller codes \cite{HP,Wei}, binary Kasami codes \cite{HK1}, Melas and dual Melas codes \cite{VV2}, BCH codes and their duals \cite{CC,C,DSV,FTW,MPP,SC,VV1,VV4,VV3,V2}, trace codes \cite{CMN,GO,SV,VV5,VV6}, product codes \cite{HK3,MW,P,S,SW,WY} and algebraic geometry codes \cite{BM,BLV,DFGL,HTV,HK4,M,MR,YKS}, etc. However, generally speaking, computing the GHWs of linear codes is difficult, and the complete weight hierarchy is known for only a few cases (see for example \cite{BM,HK1,HK2,HTV,P,V1,Wei,WY,YLFL}). Interested readers may refer to Section 7.10 of the excellent textbook \cite{HP1} for a brief introduction to GHWs, and to \cite{TV} for a comprehensive survey of GHWs via a geometric approach.

In a recent interesting paper \cite{YLFL}, the authors employed some new ideas from number theory to study the GHWs of irreducible cyclic codes and obtained the weight hierarchy for several cases. Their results extend some previous works \cite{HK2,V1}. Inspired by this work, in this paper we study the GHWs of a family of reducible cyclic codes which were introduced in \cite{YXDL} where the weight distribution was obtained in several cases. This family of cyclic codes may be characterized as having arbitrary number of nonzeroes and contain as special cases many subfamilies of cyclic codes whose weight distribution was investigated in the literature, and in particular contain the irreducible cyclic codes considered in \cite{YLFL}. We determine the weight hierarchy of this family of cyclic codes in several cases by extending the ideas of \cite{YLFL} into higher dimension and by employing some interesting combinatorial arguments.

We organize the paper as follows. In Section \ref{sec2}, we introduce the family of reducible cyclic codes we are interested in and state the main results Theorems \ref{sec3:main1} and \ref{sec3:main2}. In Section \ref{sec3} we briefly introduce some basic concepts and notions which will be used in the proofs. In Section \ref{sec4} we derive an expression of GHWs in terms of Gauss periods. Sections \ref{sec:p1} and \ref{sec:p2} are devoted to the proofs of parts (i) and (ii) of Theorem \ref{sec3:main1} respectively. We prove Theorem \ref{sec3:main2} in Section \ref{sec:p3}. Section \ref{sec6} concludes the paper.

\section{Main results}\label{sec2}


Let $q=p^s$, $Q=q^m$ where $p$ is a prime number, $s,m$ are positive integers. Let $\ga$ be a primitive element of the finite field $\FQ$. Assume the following three assumptions:
\begin{itemize}
\item[i)] $e \mid (Q-1)$, $a \not\equiv 0 \pmod{Q-1}$, $e \ge t \ge 1$;
\item[ii)] For $1 \le i \le t$, $a_i \equiv a+\frac{Q-1}{e}\De_i \pmod{Q-1}$ where $\De_i \not\equiv \De_j \pmod{e}$ for $i \ne j$ and $\gcd(\De_2-\De_1,\ldots,\De_t-\De_1,e)=1$;
\item[iii)] $\deg h_{a_i}(x)=m$ for $1 \le i \le t$ and $h_{a_i}(x) \ne h_{a_j}(x)$ for $i \ne j$. Here $h_a(x)$ denotes the minimal polynomial of $\ga^{-a}$ over $\Fq$.
\end{itemize}
Let us define
\begin{align*}
\de&=\gcd(Q-1,a_1,a_2,\ldots,a_t),\\
n&=\frac{Q-1}{\de}, \qquad N=\gcd\left(\frac{Q-1}{q-1},ae\right).
\end{align*}
It is clear that $\de \mid \frac{Q-1}{e}$ and hence
\begin{equation}\label{con1}
e\de \mid N(q-1).
\end{equation}
Under the above three assumptions, we define $\cC$ to be the cyclic code of length $n$ over $\Fq$ with parity check polynomial given by $\prod_{i=1}^t h_{a_i}(x)$, where $a_i$'s are specified according to the assumptions. It is known that $\cC$ is an $[n,tm]$ cyclic code with $t$ nonzeroes.

We remark that this family of cyclic codes were originally defined in \cite{YXDL} in which the weight distribution was obtained in several cases. There is an abundance of $\cC$ in the family, given the flexibility of parameters $e,t,\De_i$'s. In particular, it was observed in \cite[Lemma 6]{YXDL} that Assumption iii) always holds true if $N \le \sqrt{Q}$, which is the case in this paper because we only treat $N=1,2$.
The main results of the paper are as follows:
\begin{theorem} \label{sec3:main1}
Let $\cC$ be the cyclic code defined as above with $e=t \ge 1$. Let $d_r:=d_r(\cC)$ be the $r$-th GHW of $\cC$.
\begin{itemize}
\item[(i).] If $N=1$, then we have
\begin{equation*}
d_r=\frac{q^m-1}{\de}\left(1-\frac{s}{t}\right)-\frac{q^{(t-s)m-r}-1}{t\de}, \quad \mbox{ if } (t-s-1)m < r \le (t-s)m,
\end{equation*}
where $0 \le s \le t-1$.
\item[(ii).] If $N=2$, then we have
\begin{equation*}
d_r=\begin{cases}
      \frac{q^m-1}{\de}(1-\frac{s}{t})-\frac{1}{t\de}(q^{(t-s-\frac12)m-r}+1)(q^{\fm}-1) & \text{if $(t-s-1)m < r \le \left(t-s-\frac12\right)m$},\\
      \frac{q^m-1}{\de}\left(1-\frac{s}{t}\right)-\frac{2}{t\de}(q^{(t-s)m-r}-1) & \text{if $\left(t-s-\frac12\right)m < r \le (t-s)m$},
      \end{cases}
\end{equation*}
where $0 \le s \le t-1$.
\end{itemize}
\end{theorem}

\begin{theorem} \label{sec3:main2}
Let $\cC$ be a cyclic code defined as above with $e>t \ge 1$ and $N=1$. Suppose $\left\{\De_1 \pmod{e},\ldots, \De_t \pmod{e}\right\}$ is an arithmetic progression. Let $d_r:=d_r(\cC)$ be the $r$-th GHW of $\cC$. Then
\begin{equation*}
d_r=\begin{cases}
      \frac{q^m-1}{\de}(1-\frac{t-1}{e})-\frac{e-t+1}{e\de}(q^{m-r}-1) & \text{if $1 \le r \le m$},\\
      \frac{q^m-1}{\de}(1-\frac{s}{e})-\frac{q^{(t-s)m-r}-1}{e\de} & \text{if $(t-s-1)m < r \le(t-s)m$},
      \end{cases}
\end{equation*}
where $0 \le s \le t-2$.
\end{theorem}
We remark that when $t=1$, Theorems \ref{sec3:main1} and \ref{sec3:main2} reduce to Corollary 3.2 and Theorem 4.1 of \cite{YLFL} respectively. When $t \ge 2$, the results are new.

It looks likely that the weight hierarchy of $\cC$ can still be determined when $e>t \ge 1$ and $N=2$, however, for a complete answer, the computation becomes quite complicated. We may come back to this problem in the future.

\section{Preliminaries}\label{sec3}

In this section, we introduce some preliminary results which will be used later in the paper.

\subsection{Cyclic Codes}

Let $\cC$ be an $[n,k]$ linear code over $\Fq$. Assume that $(n,q)=1$. $\cC$ is called cyclic, if in additional to be linear, $\cC$ satisfies the property that the cyclic shift  $(c_{n-1},c_0,\ldots,c_{n-2}) \in \cC$ whenever $(c_0,c_1,\ldots,c_{n-1}) \in \cC$. For such a cyclic code $\cC$, each codeword $c=(c_0,\ldots,c_{n-1})$ can be associated with a polynomial $\sum_{i=0}^{n-1} c_ix^i$ in the principal ideal ring $R_n:=\Fq[x]/(x^{n}-1)$. Under this correspondence, $\cC$ can be identified with an ideal of $R_n$. Hence, there is a unique polynomial $g(x) \in \Fq[x]$ with $g(x) \mid x^n-1$ such that $\cC=(g(x))R_n$. The $g(x)$ is called the {\em generator polynomial} of $\cC$, and $h(x)=\frac{x^n-1}{g(x)}$ is called the {\em parity check polynomial} of $\cC$. When $R_n$ is specified, a cyclic code is uniquely determined by either the generator polynomial or parity check polynomial. $\cC$ is said to have $i$ {\em zeroes} if its generator polynomial can be factorized into a product of $i$ irreducible polynomials over $\Fq$. When the dual code $\cC^{\perp}$ has $i$ zeroes, we call $\cC$ a cyclic code with $i$ {\em nonzeroes}.

\subsection{Group Characters, Gauss Sums and Gauss Periods}

Let $q=p^s$ where $p$ is a prime number. The {\it canonical additive character} $\psi_q$ of $\Fq$ is given by
\begin{align*}
  \psi_q: \Fq & \longrightarrow \bb{C} \\
         x  & \mapsto \zp^{\Tqp(x)},
\end{align*}
where $\zp=\exp\left(2 \pi \sqrt{-1} /p\right)$ is a primitive $p$-th root of unity of $\bb{C}$, and $\Tqp$ is the trace function from $\Fq$ to $\Fp$. If $Q$ is a power of $q$, by the transitivity of trace functions we have $\psi_Q=\psi_q \circ \TQq$.

Let $\chi: \Fq^*=\Fq \setminus \{0\} \to \bb{C}$ be a multiplicative character, that is, $\chi(xy)=\chi(x)\chi(y)$ for any $x,y \in \Fq^*$. We may extend the definition of $\chi$ to $\Fq$ by setting $\chi(0)=0$. The corresponding {\it Gauss sum} $G(\chi)$ is defined by
$$
G(\chi)=\sum_{x \in \Fq}\chi(x)\psi_q(x),
$$
where $\psi$ is the canonical additive character. If the order of $\chi$ equals two, $\chi$ is called the quadratic character of $\Fq$ and the corresponding Gauss sum the quadratic Gauss sum. The explicit values of quadratic Gauss sums are known and are recorded as follows.

\begin{lemma}{\rm \cite[Theorem 5.15]{LN}}\label{lemmags}
Let $q=p^s$ and $\chi$ be the quadratic character of $\Fq$ (hence $p$ is odd). Then
$$
G(\chi)=\begin{cases}
      (-1)^{s-1}\sqrt{q} & \text{if $p\equiv 1 \bmod 4$},\\
      (-1)^{s-1}\left(\sqrt{-1}\right)^s\sqrt{q} & \text{if $p\equiv 3 \bmod 4$}.
\end{cases}
$$
\end{lemma}

Let $\ga$ be a primitive element of $\Fq$. For $N \mid (q-1)$, denote by $\langle \ga^N \rangle$ the multiplicative subgroup generated by $\ga^N$. Then for any $0 \le i \le N-1$, $\CNq_i=\ga^i \langle \ga^N \rangle$ is called the $i$-th cyclotomy class of $\Fq$. For any $a \in \Fq$, we define the {\it Gauss period} $\etaNq_a$ as
$$
\etaNq_a=\sum_{x \in \CNq_0}\psi_q(ax).
$$

\section{An expression of the GHWs} \label{sec4}

Let $\cC$ be the cyclic code defined in Introduction. By Delsarte's Theorem \cite{Del}, elements of $\cC$ can be represented uniquely by $c(\ux)=(c_i(\ux))_{i=1}^n$ where $\ux=(x_1,x_2,\ldots,x_t)$ runs over the set $\FQ^t$ and
$$
c_i(\ux)=\TQq\left(\sum_{j=1}^t x_j\ga^{a_ji}\right), \quad 1 \le i \le n.
$$
In other words, the map
\begin{align*}
  \Psi: \FQ^t & \longrightarrow \cC, \\
         \ux  & \mapsto c(\ux)
\end{align*}
is an isomorphism between two $\Fq$-vector spaces $\FQ^t$ and $\cC$, hence induces a 1-1 correspondence between $r$-dimensional $\Fq$-subspaces of $\FQ^t$ and $r$-dimensional subcodes of $\cC$ for any $1 \le r \le tm$. For any $\Fq$-vector space $M$, denote by $\left[{M \atop r}\right]$ the set of $r$-dimensional $\Fq$-subspaces of $M$.

For any $H_r \in \subsp$, define
$$
N(H_r)=|\{i : 1 \le i \le n, c_i(\ub)=0, \forall \ub \in H_r\}|,
$$
and for any $1 \le r \le tm$, define
$$
N_r=\max\left\{N(H_r) \mid H_r \in \subsp\right\}.
$$
Since $\Psi$ is an isomorphism, by definition, the $r$-th GHW of $\cC$ can be expressed as
$$
d_r:=d_r(\cC)=n-N_r.
$$
Let $\ul{\ep_1},\ldots,\ul{\ep_r}$ be an $\Fq$-basis of $H_r$. Then,
$$
c_i(\ub)=0, \forall \ub \in H_r \Longleftrightarrow c_i(\ul{\ep_j})=0, \forall 1 \le j \le r.
$$
As before we use $\psi_Q$ (resp. $\psi_q$) to denote the canonical additive character of $\FQ$ (resp. $\Fq$). By using orthogonality of $\phi_q$ we have
\begin{align*}
N(H_r)&=\sum_{i=1}^n \left\{\frac{1}{q} \sum_{x_1 \in \Fq} \psi_q(x_1c_i(\ul{\ep_1})) \right\} \cdots \left\{ \frac{1}{q} \sum_{x_r \in \Fq} \psi_q(x_rc_i(\ul{\ep_r})) \right\}\\
      &=\frac{1}{q^r}\sum_{i=1}^n \sum_{x_1,\ldots,x_r \in \Fq} \psi_q(c_i(x_1\ul{\ep_1}+\cdots+x_r\ul{\ep_r}))\\
      &=\frac{1}{q^r}\sum_{i=1}^n \sum_{\ub \in H_r}\psi_q(c_i(\ub))=\frac{n}{q^r}+\frac{1}{q^r}\sum_{i=1}^n \sum_{\ub \in H_r^*}\psi_q(c_i(\ub)),
\end{align*}
where $H_r^*=H_r \sm \zs$. Recall that $\ga^{a_j}=\ga^{a+\frac{Q-1}{e}\De_j}$. By setting $\be_j=\ga^{\frac{Q-1}{e}\De_j}$, we have $\ga^{a_j}=\ga^a\be_j$. Thus,
\begin{align*}
(q-1)N(H_r)&=\frac{n(q-1)}{q^r}+\frac{1}{q^r}\sum_{i=1}^n \sum_{\ub \in H_r^*}\sum_{x \in \Fq^*}\psi_q(c_i(x\ub))\\
&=\frac{n(q-1)}{q^r}+\frac{1}{q^r}\sum_{i=1}^n \sum_{\ub \in H_r^*}\sum_{x \in \Fq^*}\psi_q\left(x\TQq\left(\sum_{j=1}^tb_j\ga^{a_ji}\right)\right)\\
&=\frac{n(q-1)}{q^r}+\frac{1}{q^r}\sum_{\ub \in H_r^*}\sum_{i=1}^n\sum_{x \in \Fq^*}\psi_Q\left(x\ga^{ai}\sum_{j=1}^tb_j\be_j^{i}\right).
\end{align*}
Noting that $e \mid \frac{Q-1}{\de}=n$, any $1 \le i \le n$ can be expressed as $i=ej+h$ with $0 \le j \le \frac{n}{e}-1$ and $1 \le h \le e$. Thus, the second term on the right can be written as
$$
\frac{1}{q^r}\sum_{\ub \in H_r^*}\sum_{j=0}^{\frac{n}{e}-1}\sum_{h=1}^e\sum_{x \in \Fq^*}\psi_Q\left(x\ga^{aej}\ga^{ah}\sum_{j=1}^tb_j\be_j^{h}\right).
$$
Noting that $\Fq^*=\left\langle \ga^{\frac{Q-1}{q-1}} \right\rangle$ and $N=\gcd\left(\frac{Q-1}{q-1},ae\right)$, it is easy to check that
$$
\left\{x\ga^{aej} \mid x \in \Fq^*, 0 \le j \le \frac{n}{e}-1\right\}=\frac{(q-1)N}{e\de}*\CN_0,
$$
where we use $l*A$ to denote a multiset whose elements are those of the set $A$, each repeating $l$ times. Consequently, we have
\begin{align*}
& \frac{1}{q^r}\sum_{\ub \in H_r^*}\sum_{j=0}^{\frac{n}{e}-1}\sum_{h=1}^e\sum_{x \in \Fq^*}\psi_Q\left(x\ga^{aej}\ga^{ah}\sum_{j=1}^tb_j\be_j^{h}\right)\\
=& \frac{(q-1)N}{e\de q^r} \sum_{\ub \in H_r^*}\sum_{h=1}^e\sum_{y \in \CN_0}\psi_Q\left(y\ga^{ah}\sum_{j=1}^tb_j\be_j^{h}\right).
\end{align*}
Setting $g=\ga^a$, we have
$$
(q-1)N(H_r)=\frac{(q-1)n}{q^r}+\frac{(q-1)N}{e\de q^r} \sum_{\ub \in H_r^*}\sum_{h=1}^e\sum_{y \in \CN_0}\psi_Q\left(yg^{h}\sum_{j=1}^tb_j\be_j^{h}\right).
$$
That is,
\begin{equation}\label{eqn1}
N(H_r)=\frac{N}{e\de q^r} \sum_{\ub \in H_r}\sum_{h=1}^e \etaNg.
\end{equation}

Therefore, to compute the GHWs for the cyclic code $\cC$, it suffices to determine the maximal value of the above sum of the Gauss periods for all $H_r \in \subsp$. The problem is difficult in general. However, when the Gauss periods in the sum take only a few values, it is hopeful to determine the GHWs completely. Thus, the value $N$ is one of the key points in the computation and we will consider below the simplest cases where $N \in \{1,2\}$. We first deal with the case that $e=t$.

\section{Proof of (i) of Theorem \ref{sec3:main1}}\label{sec:p1}

Recall that $\be_j=\ga^{\frac{Q-1}{t}\De_j}$ for $1 \le j \le t$ and $g=\ga^a$. Define $\be=\ga^{\frac{Q-1}{t}}$. Since $\be_j^t=1$ for each $1 \le j \le t$ and $\be_j$'s are distinct, we can assume without loss of generality that $\be_j=\be^{j}$, $1 \le j \le t$. We may make a change of variables $\ub=(b_1,\ldots,b_t) \mapsto \uy=(y_1,\ldots,y_t)$ by $y_i=g^i\sum_{j=1}^tb_j\be^{ij},1 \le i \le t$. It is easy to see that this defines an $\Fq$-isomorphism $\phi: \FQ^t \longrightarrow \FQ^t$. Thus Equation (\ref{eqn1}) can be written as
\begin{equation*}
N(H_r)=\frac{N}{t\de q^r} \sum_{\uy \in \phi(H_r)}\sum_{h=1}^t \etaNy.
\end{equation*}
Define
\begin{equation*}
\tN(H_r)=\frac{N}{t\de q^r} \sum_{\uy \in H_r}\sum_{h=1}^t \etaNy.
\end{equation*}
Clearly,
$$
\max\left\{N(H_r) \mid H_r \in \subsp\right\}=\max\left\{\tN(H_r) \mid H_r \in \subsp\right\}.
$$
Hence, we can focus on the determination of
\begin{align}
N_r &=\max\left\{\tN(H_r) \mid H_r \in \subsp\right\} \notag\\
    &= \frac{N}{t\de q^r} \max\left\{F(H_r) \mid H_r \in \subsp\right\} \label{eqn2},
\end{align}
where
\begin{equation}\label{eqn3}
F(H_r)=\sum_{\uy \in H_r}\sum_{h=1}^t \etaNy.
\end{equation}

Now we are ready to prove (i) of Theorem \ref{sec3:main1}.
\subsection{Proof of (i) of Theorem \ref{sec3:main1}: $N=1$}
\begin{proof}
Since $N=1$, by $(\ref{con1})$, we have $e\de \mid (q-1)$ and hence $e=t \le q-1$. Recall that
$$
\etao_a=\begin{cases}
  Q-1 & \text{if $a=0$},\\
  -1  & \text{if $a \ne 0$}.
\end{cases}
$$
For any $H_r \in \subsp$ and any $1 \le h \le t$, define
$$
V_h=\underbrace{\FQ\times\cdots\times\FQ}_{h-1} \times \underset{h}{\zs}\times \underset{h+1}{\FQ} \times \ldots \times \underset{t}{\FQ}
$$
and
$$
v_h:=v_h(H_r)=\dim_{\Fq}(H_r \cap V_h).
$$
Then, by Equation (\ref{eqn3}),
\begin{align}
F(H_r)&=\sum_{h=1}^t (|H_r \cap V_h|(Q-1)-(q^r-|H_r \cap V_h|)) \notag \\
      &=\sum_{h=1}^tq^{v_h}((Q-1)-(q^r-q^{v_h}))  \notag \\
      &=Q\sum_{h=1}^tq^{v_h}-tq^r. \notag
\end{align}
Hence, for any $r$-dimensional subspace $H_r$, the tuple $(v_1(H_r),\ldots,v_t(H_r))$ completely determines $F(H_r)$.

To find $\max\left\{F(H_r): H_r \in \subsp\right\}$, we may assume that $v_1\ge v_2\ge\cdots\ge v_t$, since the ordering of $v_h$'s does not affect the value. We perform the following procedure successively: first we find $H_r$'s that maximize the value $v_1=v_1(H_r)$; once $v_1,\ldots,v_i$ are determined for $i \ge 1$, then among the $H_r$'s we find the ones that maximize $v_{i+1}$. Since $t \le q-1$, it is easy to see that this procedure will produce the desired $H_r$ which maximizes $F(H_r)$.

Now suppose $(t-s-1)m<r\le(t-s)m$ for some $s$ in the range $0 \le s \le t-1$. Noting that
$$
\dim_{\Fq}\left(\bigcap_{i=1}^s V_i\right)=(t-s)m,\quad \dim_{\Fq}\left(\bigcap_{i=1}^{s+1} V_i\right)=(t-s-1)m,
$$
we may take $v_1=\ldots=v_s=r$, which are obviously maximal. This means that $H_r \subset \bigcap_{i=1}^s V_i$. To make $v_{s+1}$ maximal, we shall take $\bigcap_{i=1}^{s+1} V_i \subset H_r \subset \bigcap_{i=1}^s V_i$, so that $v_{s+1}=(t-s-1)m$. Under this configuration, $H_r$ is of the form
$$
H_r=\underbrace{\zs\times\cdots\times\zs}_s \times \underset{s+1}{H}\times \underset{s+2}{\FQ} \times \ldots \times \underset{t}{\FQ},
$$
where $H \subset \FQ$ is any $\Fq$-vector space of dimension $r-(t-s-1)m$, and we find
\begin{equation*}
v_h=v_h(H_r)=\left\{\begin{array}{cl}
          r & 1 \le h \le s, \\
          (t-s-1)m & h=s+1, \\
          r-m & s+2 \le h \le t.
          \end{array}\right.
\end{equation*}
Therefore we obtain
$$
\max\left\{F(H_r) \mid H_r \in \subsp\right\}=Q\left(sq^r+q^{(t-s-1)m}+(t-s-1)q^{r-m}\right)-tq^r,
$$
and
\begin{align*}
d_r&=n-\frac{1}{t\de q^r}\max\left\{F(H_r) \mid H_r \in \subsp\right\} \\
   &=\frac{q^m-1}{\de}(1-\frac{s}{t})-\frac{q^{(t-s)m-r}-1}{t\de}.
\end{align*}
This completes the proof of (i) of Theorem \ref{sec3:main1}.
\end{proof}

\section{Proof of (ii) of Theorem \ref{sec3:main1}: $N=2$} \label{sec:p2}


For $t=e \ge 1$ and $N =2$, the situation is more difficult. We first adopt a new strategy that works for any $N \ge 2$. This strategy was inspired by \cite{YLFL}.

\subsection{A new strategy}

Let $\langle \cdot, \cdot \rangle: \FQ^t \times \FQ^t \to \Fq$ be the non-degenerate bilinear form given by
$$
\langle \ux ,\uy\rangle=\TQq\left(\sum_{i=1}^t x_iy_i\right),\quad \forall \, \ux=(x_1,\ldots,x_t), \uy=(y_1,\ldots,y_t) \in \FQ^t.
$$
Then for any $\Fq$-subspace $H$ of $\FQ^t$, we define
$$
\Hp=\{y \in \FQ^t \mid \lan \ux, \uy \ran=0, \forall \ux \in H \}.
$$
We have the following lemma.

\begin{lemma}
Suppose $H \in \subsp$. Then $\dim_{\Fq}\Hp =tm-r$ and
$$
\frac{1}{q^{tm-r}}\sum_{\uy \in \Hp} \psi_q \left(\lan \ux, \uy\ran\right)=\begin{cases}
1 & \text{if $\ux \in H$},\\
0 & \text{if $\ux \not\in H$}.
\end{cases}
$$
\end{lemma}
\begin{proof}
Let $A=\frac{1}{q^{tm-r}}\sum_{\uy \in \Hp} \psi_q (\lan \ux, \uy\ran)$. If $\ux \in H$, then $\lan \ux, \uy \ran=0$, $\forall \uy \in \Hp$. Hence, $A=1$. If $\ux \not\in H$, then there exists $\uy \in \Hp$ such that $\lan \ux, \uy \ran\ne 0$. In particular, there exists $\ul{y_0} \in \Hp$, such that $\psi_q(\lan \ux, \ul{y_0}\ran)\ne 1$. Thus,
\begin{align*}
A\cdot \psi_q(\lan \ux, \ul{y_0} \ran) &=\frac{1}{q^{tm-r}}\sum_{\uy \in \Hp} \psi_q(\lan \ux, \uy\ran+\lan \ux, \ul{y_0} \ran) \\
&= \frac{1}{q^{tm-r}}\sum_{\uy \in \Hp} \psi_q (\lan \ux, \uy+\ul{y_0}\ran)= A.
\end{align*}
Hence $A=0$. This completes the proof.
\end{proof}

By the above lemma, we can compute $F(H_r)$ in Equation (\ref{eqn3}) as follows.
\begin{align*}
F(H_r)&=\sum_{\uy \in H_r}\sum_{h=1}^t \etaNy =\sum_{h=1}^t\sum_{\uy \in \FQ^t} \etaNy \frac{1}{q^{tm-r}} \sum_{\ux \in \Hrp} \psi_q (\lan \ux, \uy \ran) \\
      &=\frac{1}{q^{tm-r}}\sum_{h=1}^t\sum_{\ux \in \Hrp}\sum_{z \in \CN_0}\sum_{\uy \in \FQ^t}\psi_Q(zy_h) \psi_q \left(\TQq\left(\sum_{i=1}^{t} x_iy_i\right)\right)\\
      &=\frac{1}{q^{tm-r}}\sum_{h=1}^t\sum_{\ux \in \Hrp}\sum_{z \in \CN_0}\sum_{y_1,\ldots,y_t \in \FQ}   \psi_Q \left(zy_h+\sum_{i=1}^{t} x_iy_i\right)\\
      &=q^r\sum_{h=1}^t\sum_{\ux \in \Hrp}\sum_{\substack{z \in \CN_0 \\ z+x_h=0 \\x_i=0, \forall \, i \ne h} }1.
\end{align*}

For any $H_r \in \subsp$ and $1\le h \le t$, define
$$
W_h:=W_h(N)=\underbrace{\zs \times \cdots \times \zs}_{h-1} \times (-\underset{h}{\CN_0}) \times \underset{h+1}{\zs} \times \cdots \underset{t}{\zs}
$$
and
$$
U_h:=U_h(H_r)=\Hrp \bigcap \biggl( \underbrace{\zs \times \cdots \times \zs}_{h-1} \times \underset{h}{\FQ} \times \underset{h+1}{\zs} \times \cdots \underset{t}{\zs} \biggr).
$$
Then we have
\begin{align}
F(H_r)    &=q^r\sum_{h=1}^t |\Hrp \cap W_h|=q^r\sum_{h=1}^t |U_h \cap W_h|.  \label{eqn6}
\end{align}

Therefore for $t=e \ge 1$ and $N \ge 2$, to compute the GHWs, it suffices to determine the maximal value of $\sum_{h=1}^t |U_h \cap W_h|$ for all $H_r \in \subsp$.

We remark that since $\Hrp=\bigoplus_{h=1}^t U_h$, the spaces $\Hrp$ and $H_r$ can be uniquely recovered from $(U_1,\ldots,U_t)$, where $U_h$ is a subspace of $\underbrace{\zs \times \cdots \times \zs}_{h-1} \times \underset{h}{\FQ} \times \underset{h+1}{\zs} \times \cdots \underset{t}{\zs}$. Setting $r'=tm-r$ and letting $u_h=\dimq(U_h)$, the $u_h$'s satisfy the condition that $\sum_{h=1}^t u_h=\rp$. So to find the maximal value of $\sum_{h=1}^t\left|U_h \cap W_h\right|$, we consider the following two steps:
\begin{itemize}
\item[1).] for each eligible tuple $(u_1,\ldots,u_h)$, determine the spaces $U_h$ of dimension $u_h$ such that $\left|U_h \cap W_h\right|$ is maximal for each $h$;

    \item[2).] consider all eligible tuples $(u_1,\ldots,u_h)$ and find the maximal.
\end{itemize}
Now we focus on the case that $N=2$.

\subsection{A working lemma}

Since $N=2$, $q$ is odd and $m$ is even. Note that $Q=q^m$. We need the following lemma.

\begin{lemma}\label{lemmaint}
Let $0 \le l \le m$ and $H \subset \FQ$ be an $l$-dimensional $\Fq$-subspace. Let $\ga$ be a primitive element of $\FQ$. Define a function
$$
f(l):=\begin{cases}
  q^l-1 & \text{if $0 \le l \le \frac{m}{2}$}, \\
  \frac{q^l-1}{2}+\frac{q^{\frac{m}{2}}-q^{l-\frac{m}{2}}}{2} & \text{if $\frac{m}{2} \le l \le m$}.
\end{cases}
$$
Then, for $0 \le l \le m$,
$$
\max\left\{\left|H \cap \Ct_i \right| : H \subset \FQ, \dimq(H)=l \right\}=f(l),
$$
where $i \in \{0,1\}$.
Furthermore, the subspace $H \subset \FQ$ that achieves the maximal value can be chosen as follows:
\begin{itemize}
\item[1)] If $0 \le l \le \fm$, then we can choose any $H \subset \ga^{i}\Fqm$;
\item[2)] If $\fm \le l \le m$, suppose $G(\chi)=(-1)^jq^{\fm}$ for some $j \in \{0,1\}$, where $\chi$ is the quadratic character of $\FQ$. Then we can choose $H$ satisfying $\Hp \subset \Ct_{i+j} \cup \zs$, here $\Hp$ is the orthogonal complement of $H$ with respect to the non-degenerate bilinear form $\langle \cdot,\cdot \rangle: \FQ \times \FQ \to \Fq$ given by
    $$
    \langle x,y \rangle= \TQq\left(xy\right), \forall x,y \in \FQ.
    $$
\end{itemize}
\end{lemma}
\begin{proof}
We will only prove the lemma for the case $i=0$. The case of $i=1$ is analogous.

If $0 \le l \le \fm$, then $\Fqm^* \subset \Ct_0$. Note that $\dimq\left(\Fqm\right)=\frac{m}{2}$. For any $H \subset \Fqm$, we have $|H \cap \Ct_0|=q^l-1$, which is clearly the maximal value.

If $\fm \le l \le m$,  then
\begin{align*}
|H \cap \Ct_0|&=\sum_{a \in H \sm \zs} \frac12 (1+\chi(a))
=\frac{q^l-1}{2}+\frac12\sum_{a \in H}\chi(a).
\end{align*}
Note that
\begin{align*}
\sum_{a \in H} \chi(a)&=\sum_{a \in \FQ}\frac{\chi(a)}{q^{m-l}}\sum_{b \in \Hp} \psi_q(\langle a,b \rangle) \\
                      &=\frac{1}{q^{m-l}}\sum_{b \in \Hp} \chi(b) \sum_{a \in \FQ} \chi(ab)\psi_q(\langle a,b \rangle)\\
                      &=\frac{G(\chi)}{q^{m-l}} \sum_{b \in \Hp}\chi(b).
\end{align*}
Since $\dimq(\Hp)=m-l \le \fm$, we can choose $\Hp$ satisfying $\Hp \subset \Ct_j \cup \zs$. Consequently, we have $\chi(b)=(-1)^j$, $\forall b \in \Hp \sm \zs$. Hence,
$$\sum_{a \in H} \chi(a)=\frac{q^{\fm}}{q^{m-l}}\sum_{b \in \Hp \sm\zs} 1=q^{\fm}-q^{l-\fm}.$$
This is clearly the largest value of $\sum_{a \in H} \chi(a)$. Therefore,
$$
\left|H \cap \Ct_0\right|=\frac{q^l-1}{2}+\frac{q^{\frac{m}{2}}-q^{l-\frac{m}{2}}}{2}
=f(l)$$
is the maximal value when $\fm \le l \le m$. This finishes the proof of Lemma \ref{lemmaint}.
\end{proof}

\subsection{Proof of (ii) of Theorem \ref{sec3:main1}}
We can now start the proof of (ii) of Theorem \ref{sec3:main1}. Recall that $H_r$ can be recovered from $(U_1,\ldots,U_t)$ where $U_h$'s are subspaces of dimension $u_h$ such that $\sum_{h=1}^t u_h=\rp=tm-r$. For any eligible fixed tuple $(u_1,\ldots,u_t)$, Lemma~\ref{lemmaint} tells us how to choose $(U_1,\ldots,U_t)$ so that $|U_h \cap W_h|$ attains the maximal value $f(u_h)$ for each $1 \le h \le t$. Hence it remains to determine which tuple $(u_1,\ldots,u_t)$ will give the maximal value $\sum_{h=1}^t f(u_h)$, here $(u_1,\ldots,u_t)$ satisfies the condition that $0 \le u_h \le m, \forall h$ and $\sum_{h=1}^t u_h = \rp$. We may assume without loss of generality that $u_1 \ge \cdots \ge u_t$.

The following lemma is useful to determine the maximal value of $\sum_{h=1}^t f(u_h)$. We define a finite set
$\cL=\{(l_1,\ldots,l_t) \mid 0 \le l_1,\ldots,l_t \le m, l_1 \ge l_2 \ge \cdots \ge l_t\}$ and for each $0 \le s \le tm$ a subset $\Ls=\{(l_1,\ldots,l_t) \in \cL \mid \sum_{i=1}^tl_i=s\}$. We can define a partial order on $\Ls$ as follows: for any $\ul{l}=(l_1,\ldots,l_t), \ul{l^{\pr}}=(l^{\pr}_1,\ldots,l^{\pr}_t) \in \Ls$, we say $\ul{l} \su \ul{l^{\pr}}$ if there is an index $i$ $(1 \le i \le t)$, such that $l_j=l^{\pr}_j$, $\forall 1 \le j \le i-1$ and $l_i>l^{\pr}_i$. It is easy to see that $\su$ gives a total order on $\Ls$.

\begin{lemma}\label{lemmaord}
For any $\ul{l}, \ul{l^{\pr}} \in \Ls$, if $\ul{l} \su \ul{l^{\pr}}$, then $\sum_{h=1}^tf(l_h) \ge \sum_{h=1}^tf(l^{\pr}_h)$.
\end{lemma}
\begin{proof}
For any $\ul{l}=(l_1,\ldots,l_t) \in \Ls$, define $f(\ul{l})=\sum_{h=1}^t f(l_h)$. For simplicity define $l_0:=m$. Suppose $l_i<l_{i-1}$ and $l_j \ge 1$ for some $i,j$ such that $1\le i<j\le t$. Then we can define an operation $S_{i,j}$ on $\ul{l}$ as follows:
$$
S_{i,j}(\ul{l})=(l_1,\ldots,l_{i-1},l_{i}+1,l_{i+1},\ldots,l_{j-1},l_{j}-1,l_{j+1},\ldots,l_t).
$$
Clearly $S_{i,j}(\ul{l}) \in \Ls$ and $S_{i,j}(\ul{l}) \su \ul{l}$. We claim that $f\left(S_{i,j}(\ul{l})\right) \ge f(\ul{l})$.

It suffices to prove that
\begin{eqnarray} f(l_i+1)-f(l_i) \ge f(l_j)-f(l_j-1), \quad l_i \ge l_j \ge 1. \label{eqn77}
\end{eqnarray}
For $1 \le l \le m$, define a function $g(l)=f(l)-f(l-1)$. A routine computation shows easily that $g$ is an increasing function. Therefore, the above inequality (\ref{eqn77}) holds. Hence the claim is proved.

Finally, Lemma \ref{lemmaord} is proved by realizing that for any $\ul{l} \su \ul{l^{\pr}}$, $\ul{l}$ can be obtained from $\ul{l^{\pr}}$ by a series of operations $S_{i,j}$ for some $i,j$'s.
\end{proof}

Armed with Lemma \ref{lemmaord}, we can prove (ii) of Theorem \ref{sec3:main1} as follows: suppose $(t-s-1)m<r \le (t-s)m$ for some $0 \le s \le t-1$, then $sm \le \rp =tm-r<(s+1)m$. We can choose an $H_r \in \subsp$, such that
$$
\Hrp=\underbrace{\FQ \times \cdots\times \FQ}_{s}\times \underset{s+1}{T} \times \underset{s+2}{\zs}\times\cdots\times\underset{t}{\zs},
$$
where $T$ is an $(\rp-sm)$-dimensional $\Fq$-subspace of $\FQ$. For $1 \le h \le t$,
$$
U_h=\Hrp \cap ( \underbrace{\zs \times \cdots \times \zs}_{h-1} \times \underset{h}{\FQ} \times \underset{h+1}{\zs} \times \cdots \underset{t}{\zs})
$$
is of the form
$$
U_h=\underbrace{\zs \times \cdots \times \zs}_{h-1} \times \underset{h}{Y_h} \times \underset{h+1}{\zs} \times \cdots \underset{t}{\zs},
$$
where $Y_h$ is an $\Fq$-subspace of $\FQ$. We can further require that for each $1 \le h \le t$, $\left|Y_h \cap (-\Ct_0)\right|$ achieves the maximal value as discussed in Lemma~\ref{lemmaint}. Therefore, by Equation (\ref{eqn6}), $F(H_r)=\sum_{h=1}^t f(u_h)$, where $u_1=\cdots=u_s=m$, $u_{s+1}=\rp-sm=(t-s)m-r$ and $u_{s+2}=\cdots=u_t=0$. Clearly $\ul{u}=(u_1,\ldots,u_t) \in \Lrp$ is the maximum according to the order $\su$. Then by Lemma~\ref{lemmaord}, the value $F(H_r)$ is maximal, namely, this $H_r$ is the $r$-dimensional subspace that maximizes $F(H_r)$ for all $H_r \in \subsp$. By Equation (\ref{eqn2}) and Lemma~\ref{lemmaint}, we obtain
$$
N_r=\frac{2}{t\de}\left[\frac{s(q^m-1)}{2}+f((t-s)m-r)\right].
$$
This completes the proof of (ii) of Theorem \ref{sec3:main1}.

\section{Proof of Theorem \ref{sec3:main2}}\label{sec:p3}

In this section, we consider the case $e>t$ and $N=1$. If $t=1$, the code $\cC$ is irreducible which was considered in \cite{YLFL}. So we may assume that $e>t\ge2$. Recall that for $1 \le r \le tm$ and $H_r \in \subsp$, by Equation (\ref{eqn1}) we have
\begin{equation*}
N(H_r)=\frac{1}{e\de q^r} \sum_{\ub \in H_r}\sum_{h=1}^e \etang,
\end{equation*}
where $g=\ga^a$, $\be=\ga^{\frac{Q-1}{e}}$ and $\be_j=\be^{\De_j}$ for $1 \le j \le t$.
For $1 \le h \le e$, define
\[W_h:=\left\{\ul{b}=(b_1,\ldots,b_t) \in \FQ^t: \sum_{j=1}^tb_j \beta_j^h=0\right\}.\]
Then each $W_h$ is an $\Fq$-vector space of dimension $(t-1)m$. Let $$w_h:=\dimq (H_r \cap W_h).
$$
Since $\etao_0=Q-1$ and $\etao_y=-1$ if $y \in \FQ^*$, we have
\begin{align}
N(H_r)&=\frac{1}{e\de q^r} \sum_{h=1}^e \left((Q-1)q^{w_h}-(q^r-q^{w_h})\right) \notag \\
        &=\frac{1}{e\de q^r} \left(Q\sum_{h=1}^e q^{w_h}-eq^r\right)=\frac{Q}{e\de q^r} \sum_{h=1}^e q^{w_h}-\frac{1}{\de}. \label{eqn7}
\end{align}
To find $\max \left\{N(H_r): H_r \in \subsp\right\}$, we may assume without loss of generality that $w_1 \ge w_2 \ge \ldots \ge w_e$. We make a change of variables $\ul{b}=(b_1,\ldots,b_t) \mapsto \ul{y}=(y_1,\ldots,y_t)$ given by $y_h=g^h\sum_{j=1}^tb_j \beta_j^h, 1 \le h \le t$. This clearly defines an $\Fq$-isomorphism $\phi: \FQ^t \to \FQ^t$. For $t+1 \le h \le e$, define
$y_h:=g^h \sum_{j=1}^t b_j\beta_j^h$. Since $\phi$ is an isomorphism, there exist $\lambda_{h,1}, \ldots, \lambda_{h,t} \in \FQ$ such that
\[y_h=\sum_{i=1}^t\lambda_{h,i} \, y_i, \quad t+1 \le h \le e. \]
Since $\left\{\De_1 \pmod{e},\ldots, \De_t \pmod{e}\right\}$ is an arithmetic progression, it is known that $\lambda_{h,i} \ne 0, \forall \,h,i$ where $ t+1 \le h \le e, 1 \le i \le t$ (see \cite{YXDL}).

Thus Equation (\ref{eqn7}) can be written as
\begin{align*}
N(H_r)=\frac{Q}{e\de q^r} \sum_{h=1}^e q^{\tw_h}-\frac{1}{\de},
\end{align*}
where
$$\tw_h:=\dimq \left(\wHr \cap \wWh\right) \mbox{ and } \tw_1 \ge \tw_2 \ge \ldots \ge \tw_e.$$
Here $\wHr=\phi(H_r)$ and $\wWh=\phi(W_h)$ is of the form
\[
\wWh=\underbrace{\FQ\times\cdots\FQ}_{h-1}\times
\underset{h}{\zs}\times\underset{h+1}{\FQ}\times\cdots\times\underset{t}{\FQ}, \quad 1 \le h \le t,
\]
and
\[
\wWh=\left\{(y_1,\cdots,y_t)\in\FQ^t : y_h=\sum_{i=1}^t\la_{h,i} y_i=0\right\}, \quad t+1 \le h \le e.
\]
Since $e \le q-1$, to find the $H_r$ that maximizes $N(H_r)$, similar to the case that $e=t$ and $N=1$, the first priority is to make $\tw_1$ as large as possible, once this is done, then we make $\tw_2$ as large as possible, etc, and finally we make $\tw_e$ as large as possible.

Now suppose that $(t-s-1)m <r \le (t-s)m$ for some $0 \le s \le t-1$. Since $\dimq\left(\bigcap_{h=1}^s \wWh\right)=(t-s)m$, similar to the case that $e=t$ and $N=1$, $N(H_r)$ achieves the maximal value for $H_r \in \subsp$ satisfying the property that $\bigcap_{h=1}^{s+1}\wWh \subset \wHr \subset \bigcap_{h=1}^s\wWh$. Consequently, $\wHr$ is of the form
$$
\wHr=\begin{cases}
\underbrace{\zs \times \cdots \times \zs}_{s} \times \underset{s+1}{T} \times \underset{s+2}{\FQ} \times \cdots \times \underset{t}{\FQ} & \text{ if $0 \le s \le t-2$,} \\
\underbrace{\zs \times \cdots \times \zs}_{s} \times \underset{s+1}{T} & \text{ if $s=t-1$.}
\end{cases}
$$
Here $T$ is an $\Fq$-subspace of $\FQ$ with dimension $r-(t-s-1)m$.

If $0 \le s \le t-2$, it is easy to check that $\tw_1=\cdots=\tw_s=r$, $\tw_{s+1}=(t-s-1)m$ and $\tw_{s+2}=\cdots=\tw_t=r-m$. For $t+1 \le h \le e$, $\tw_h$, by definition, is the $\Fq$-dimension of the space of $(y_{s+1},\ldots,y_t)$ such that
$$
\la_{h,s+1}y_{s+1}+\la_{h,s+2}y_{s+2}+\cdots+\la_{h,t}y_t=0,
$$
where $y_{s+1} \in T$ and $y_{s+2},\ldots,y_t \in \FQ$. Since $\forall \, h,i$, $\lambda_{h,i} \ne 0$, it is easy to see $\tw_h=r-m$ for $t+1 \le h \le e$. Therefore this $H_r$ maximizes $\sum_{h=1}^eq^{\tw_h}$ whose value is given by
$$
\sum_{h=1}^eq^{\tw_h}=sq^r+q^{(t-s-1)m}+(e-s-1)q^{r-m}.
$$

If $s=t-1$, we can see that $\tw_1=\cdots=\tw_s=r$ and $\tw_{s+1}=0$. Hence for any $s+2 \le h \le e$, by the inequality $0 \le \tw_h \le \tw_{s+1}$, we also have $\tw_h=0$ for $s+2 \le h \le e$. That is, the $H_r$ maximizes $\sum_{h=1}^eq^{\tw_h}$ whose value is given by
$$
\sum_{h=1}^eq^{\tw_h}=sq^r+e-s.
$$
Then a routine computation completes the proof.

\section{Conclusion}\label{sec6}

The generalized Hamming weights (GHWs) are fundamental parameters of linear codes. They convey the structural information of a linear code and determine its performance in various applications. However, the computation of the GHWs of linear codes is difficult in general. In this paper, we study the GHWs of a family of reducible codes introduced in \cite{YXDL} and obtain the weight hierarchy in several cases. This is achieved by extending the idea of \cite{YLFL} into higher dimension and by employing some interesting combinatorial arguments. It shall be noted that these cyclic codes may have arbitrary number of nonzeroes.

\subsection*{Acknowledgments}
The research of M. Xiong was supported by RGC grant number 606211 and 609513 from Hong Kong. The research of G. Ge was supported by the National Natural Science Foundation of China under Grant No. 61171198 and Grant No. 11431003, the Importation and Development of High-Caliber Talents Project of Beijing Municipal Institutions, and Zhejiang Provincial Natural Science Foundation of China under Grant No. LZ13A010001.

\end{document}